\documentclass[11pt,reqno]{amsart}
\usepackage[english]{babel}
\usepackage{mathtools,amsmath,amssymb}
\usepackage{mathrsfs,dsfont}
\usepackage{enumitem}
\usepackage{soul,xcolor}
\usepackage{tikz-cd}
\usepackage{stmaryrd}
\usepackage{slashed}

\usepackage[pdfencoding=auto, psdextra]{hyperref}
\hypersetup{colorlinks=true,allcolors=[rgb]{0,0,0.6}}
\usepackage[hmargin=2.7cm,vmargin=2.6cm]{geometry}

\usepackage{dsfont}
\usepackage{amsthm}
\usepackage{xcolor}
\usepackage[normalem]{ulem}

\newcommand{\N}{\mathbb{N}}
\newcommand{\R}{\mathbb{R}}
\newcommand{\C}{\mathbb{C}}

\newcommand{\Id}{\mathds{1}}

\newcommand{\Hi}{\mathscr{H}}
\newcommand{\eps}{\varepsilon}
\newcommand{\Tr}{\mathrm{Tr}}
\newcommand\norm[1]{\left\lVert#1\right\rVert}
\newcommand\abs[1]{\left\lvert#1\right\rvert}
\newcommand\pure[1]{| #1 \rangle \langle #1 |}
\newcommand{\di}{\mathrm{d}}

\newcommand{\qh}{\hat{q}}
\newcommand{\ph}{\hat{p}}

\numberwithin{equation}{section}

\theoremstyle{plain} 
\newtheorem{thm}[equation]{Theorem}
\newtheorem{cor}[equation]{Corollary}
\newtheorem{lem}[equation]{Lemma}
\newtheorem{prop}[equation]{Proposition}

\theoremstyle{definition}
\newtheorem{defn}[equation]{Definition}

\theoremstyle{remark}
\newtheorem{rem}[equation]{Remark}

\begin{document}

\title[Semiclassical limit of polaron dynamics]{Semiclassical limit of polaron dynamics with cutoff
}

\author{Raphaël Gautier}
\address{Dipartimento di
  Matematica\\Politecnico di Milano\\P.zza Leonardo da Vinci 32\\20133,
  Milano\\Italy\\\& Université de Rennes\\ Campus de Beaulieu\\263, avenue du
  Général Leclerc\\35042 RENNES CEDEX\\France}
\email{raphael.gautier@univ-rennes.fr}

\date{}

\subjclass[2020]{}
\keywords{}

\begin{abstract}
We define and study the global well-posedness of a classical polaron dynamics and its derivation from the quantum dynamics of the Fröhlich Hamiltonian from the viewpoint of Wigner measures. We make a crucial use of the Gross transform, allowing us to extract a diverging constant. Finally, we discuss the difficulty of defining the classical dynamics when removing the cutoff in our setting.
\end{abstract}

\maketitle
\tableofcontents

\maketitle

\section{Introduction}
The Fröhlich Hamiltonian describes a system consisting of an electron coupled with an induced polarization medium, that can be viewed as a (quasi-)particle, called the polaron. One can interpret the Fröhlich model in the so-called strong coupling limit as a system composed of a quantum particle coupled with a semiclassical bosonic field, and in this scaling, the self-energy and the effective mass of the polaron have been extensively studied, see \cite{lieb1958,lieb1997,Sei02,Mol06,MS07,Ric14,Sei19, BS23, BS24, BM24}, and at the classical level as well, i.e. for an effective evolution described by the Landau-Pekar equations, see, for instance, \cite{LR13a,LR13b,FG17,Gri17,GSS17,LRSS19,LMS20,FRS22}. For the study of quantum corrections, see \cite{FS19,FS21}, and for a probabilistic approach on the strong coupling situation using path measures, one can refer to \cite{DV83,DS20,BP23,Sell24,BSS25}, for example.
\\ \\
In this work, we consider the full semiclassical viewpoint by considering the particle also classical instead of quantum. This will lead on the classical side to a system of PDEs describing the system in this scaling limit. To establish the derivation, we make use of the Wigner measure approach in an infinite dimensional phase-space initiated in \cite{ammari2008ahp}. The latter led to multiple results on the semiclassical limit of quantum models, such as the Nelson model \cite{ammari2014jsp,farhat2024}, its renormalized version in the mean-field limit in \cite{ammari2017sima}, or Pauli-Fierz \cite{leopold2020sima,afh22}. This measure theoretical approach has also been developed in \cite{correggi2021jems} to analyze systems describing quantum particles (or fields) interacting with semiclassical bosonic fields. The Nelson, Fröhlich, and Pauli-Fierz dynamics in this scaling are studied in \cite{CF18,correggi2021jems}, and recent work studied further properties of such systems, such as ground state energy \cite{correggi2020arxiv} and decoherence \cite{fantechi2024arxiv}.
\\ \\
A particularity of the Fröhlich Hamiltonian, shared by the Nelson model, is the singularity lying in the interaction term. By a well-known trick of Lieb and Yamazaki, the Hamiltonian can be viewed as the one emerging from a closed and bounded from below quadratic form. However, to better understand the properties of such an operator, one needs to introduce an ultraviolet cutoff, which cuts high frequencies of the interaction.
\\ \\
In \cite{gross1962}, Gross introduced a ``dressing", a unitary transform allowing to extract the divergence of the operator. In this manner, one can define the renormalized Hamiltonian as a self-adjoint operator. A complete description of its domain of has been studied in \cite{griesemer2016}. We study here the dynamics generated by the dressed and undressed Hamiltonian in the semiclassical scaling, as well as the semiclassical limit of the dressing transform, and their convergence towards globally well-posed classical PDEs. We restrict our paper to the case with UV-cutoff, and discuss the difficulty of removing the cutoff in our scaling.

\section{Settings and main results}

We consider the classical phase space $$\Hi := \Hi_p \oplus \Hi_f,$$
where 
$\Hi_p := \R^{2 d} \cong \C^d$ describes a classical particle and $\Hi_f := L^2 (\R^d, \di k)$ a classical field. The space $\Hi$ is endowed with the norm $$\norm{(q,p,\alpha)} := \abs{q}^2 + \abs{p}^2 + \norm{\alpha}_{L^2}^2,$$ and the symplectic form $$\sigma((z_1,\alpha_1),(z_2,\alpha_2)) := \mathrm{Im} \langle z_1 , z_2 \rangle_{\C^d} + 2 \mathrm{Im} \langle \alpha_1, \alpha_2 \rangle_2.$$
The Hilbert space of the quantum theory is the bosonic Fock space over the phase space $\Hi$:
$$
\Gamma_s (\Hi) := \bigoplus_{n \in \N} \Hi^{\otimes_s n} \cong L^2 (\R^d, \di q) \otimes \Gamma_s (\Hi_f).
$$
The polaron Hamiltonian is:
$$
H^\Lambda_\hbar:= \frac{\abs{\ph}^2}{2} + N_\hbar + \phi_\hbar (f^\Lambda_{\qh}) = H_\hbar^0 + \phi (f^\Lambda_{\qh}).
$$
Here $ \ph = -i \hbar \nabla_q $ is the momentum operator, $\Lambda >0$ is a fixed cutoff and 
$$f^\Lambda_{\qh} (k) := \frac{e^{-i k . \qh} \Id_{\abs{k}\leq \Lambda} }{\abs{k}^{(d-1)/2}} = e^{-i k . \qh} f^\Lambda (k)$$ is viewed as a bounded form factor $f_{q} \in L^\infty_q (L^2_k) = L_q^\infty(\Hi_f)$. 

The operator $N_\hbar = \Id_{L^2 (\R^d)} \otimes \di \Gamma_\hbar (\Id_{\Hi_f})$ is called number operator, $\di \Gamma_\hbar (\Id_{\Hi_f})$ corresponds to the second quantization of $\Id_{\Hi_f}$, and $\phi (f^\Lambda_{\qh}) = \frac{1}{\sqrt 2} ( a_\hbar(f^\Lambda_{\qh}) + a_\hbar^*(f^\Lambda_{\qh}) )$ is the Segal field operator, where $a_\hbar(f^\Lambda_{\qh}),a_\hbar^*(f^\Lambda_{\qh})$ are the $\hbar$-dependent creation and annihilation operator satisfying the following commutation relations:
$$[a_\hbar(\xi),a_\hbar^* (\eta)] = \hbar \langle \xi, \eta \rangle \Id,      \ \ \
    [a_\hbar(\xi), a_\hbar(\eta)] =  [a^*_\hbar(\xi), a^*_\hbar(\eta)]=0.
    $$

Here $\phi (f^\Lambda_{\qh})$ models the interaction between the particle and the bosonic field.
We call $$H_\hbar^0 = \frac{\abs{\ph}^2}{2} + N_\hbar $$ the free (noninteracting) Hamiltonian. \newline \newline
The operator $H_\hbar^\Lambda$ is self-adjoint on $D(\abs{\ph}^2 + N) = D(\abs{\ph}^2) \cap D(N)$ by Kato-Rellich theorem, since $f^\Lambda_{\qh} \in L^2$. However, if one tries to remove the cutoff, i.e. substituing $\Lambda = \infty$, then $f^\infty_{\qh}$ is no longer square integrable. The Hamiltonian $H^\infty_\hbar$ still remains well defined as a quadratic form on $Q(H_\hbar^0)$, thanks to the interaction $\phi (f^\infty_{\qh})$ being a  KLMN perturbation of the free Hamiltonian $H_\hbar^0$. A short proof of these statements is provided in Appendix A for completeness. In \cite{griesemer2016}, it was shown that the polaron Hamiltonian could still be defined as a norm resolvent limit, but with a domain that is disjoint from the domain of the free Hamiltonian: namely,
$$
D(H^\infty_\hbar) \cap D (\abs{\ph}^2) = \{ 0 \}.
$$
To obtain such result, a crucial point was to use a unitary transform, called Gross transform, which is rescaled in our case with the semiclassical parameter $\hbar$:
\begin{equation}
    \label{Gross}
    U_\hbar^\Lambda := e^{\frac{i}{\hbar} \phi_\hbar (i B^\Lambda_{\qh,\hbar})}, 
\end{equation}
where 
$$
B^\Lambda_{\qh,\hbar}(k) = - \frac{\Id_{\abs{k} \geq K}}{1 + \hbar \abs{k}^2/2} f^\Lambda_{\qh}(k) = e^{- i k . \qh} B_\hbar^\Lambda (k).
$$
When $\hbar =0$, we will denote $B^\Lambda_q :=B^\Lambda_{q,0}$. One notices that for all $\hbar >0$, $B^\infty_{\qh,\hbar}$ is now square-integrable. The transformed Hamiltonian, called the dressed Hamiltonian, is self-adjoint on $D(H_\hbar^0)$ and has the following form:
\begin{equation}
    \label{dressed}
\begin{aligned}
\hat{H}_\hbar^\Lambda &= U_\hbar^\Lambda H_\hbar^\Lambda (U_\hbar^\Lambda )^* - C_\hbar^\Lambda. \\
&= \frac{\abs{\ph}^2}{2} + N + \phi(f^K_{\qh}) - \frac{1}{\sqrt{2}} ( a^* (k B_{\qh,\hbar}^\Lambda) . \ph +  \ph . a(k B_{\qh,\hbar}^\Lambda) ) + \frac{\phi (k B_{\qh,\hbar}^\Lambda)^2}{2},
\end{aligned}
\end{equation}
with $C_\hbar^\Lambda = \norm{B_\hbar^\Lambda}^2 + \langle B_\hbar^\Lambda , f^\Lambda \rangle$. The latter scalar quantity diverges in the limit $\hbar \rightarrow 0$, but is not seen by the quantum evolution. It acts as a vacuum renormalization of the model, and has also effect on the classical side. Define the undressed energy functional as
    $$  \mathscr E^\Lambda : \begin{array}{ll} \Hi &\longrightarrow \R \\
        (q,p,\alpha) &\longmapsto \frac{\abs{p}^2 }{2} + \norm{\alpha}_2^2 + \sqrt{2} \mathrm{Re} \langle \alpha, f_q^\Lambda \rangle_2.
        \end{array} 
        $$
Then it is easy to see that
$$\underset{\Hi}{\inf} \ \mathscr E^\Lambda = \underset{(q,p,\alpha) \in \Hi}{\inf} \left\{ \frac{\abs{p}^2}{2} + \norm{\alpha + \frac{1}{\sqrt{2}}f_q^\Lambda}_2^2 - \frac{\norm{f^\Lambda}^2_2}{2} \right\} =   - \frac{\norm{f^\Lambda}^2_2}{2}>-\infty. $$
Since $\norm{f^\Lambda} \underset{\Lambda \rightarrow +\infty}{\longrightarrow} \infty$, the undressed classical energy is unbounded from below in the limit of cutoff removal. Define now the dressed energy functional
\begin{align*}
    \hat{\mathscr E}^\Lambda : \begin{array}{ll} \Hi &\longrightarrow \R \\
        (q,p,\alpha) &\longmapsto \frac{\abs{p}^2 }{2} + \norm{\alpha}_2^2 + \sqrt{2} \mathrm{Re} \langle \alpha, f_q^K \rangle_2 - \sqrt{2}  \mathrm{Re} \langle \alpha, (k.p) B_{q}^\Lambda \rangle_2 + (\mathrm{Re} \langle \alpha, k B_{q}^\Lambda \rangle_2)^2.
    \end{array}
\end{align*}
Then,
\begin{align*}\underset{\Hi}{\inf} \ \hat{\mathscr E}^\Lambda &= \underset{(q,p,\alpha) \in \Hi}{\inf} \left\{ \frac{\abs{p - \sqrt{2} \mathrm{Re} \langle \alpha , k B_{q}^\Lambda \rangle_2}^2}{2} + \norm{\alpha + \frac{1}{\sqrt{2}}f_q^K}_2^2 - \frac{\norm{f^K}^2_2}{2} \right\} \\ &\geq  - \frac{\norm{f^K}^2_2}{2}, 
\end{align*}
so that $\hat{\mathscr E}^\Lambda$ is bounded from below uniformly in $\Lambda$. This means that using a dressing on the quantum side, it has been possible to renormalize the classical energy, uniformly with respect to the cutoff. Making sense of those energies in the limit $\Lambda \rightarrow \infty$ is a non trivial issue, motivating us to investigate the semiclassical limit of the dressed polaron Hamiltonian with cutoff. This will be done by factorizing the analysis, i.e. studying the undressed Hamiltonian as well as the dressing transform.
\\
\\
Consider the Hamilton equations associated respectively to the undressed and dressed energy:
\begin{equation}
    \label{undressedeulerlagrange}
    \begin{cases}
        \dot q_t &= p_t, \\
        \dot p_t &= \sqrt{2}\mathrm{Re} \langle \alpha_t, i k f_{q_t}^\Lambda \rangle_2, \\
        i \dot \alpha_t &= \alpha_t + \frac{1}{\sqrt{2}} f_{q_t}^\Lambda,
    \end{cases}
\end{equation}
and
\begin{equation}
    \label{dressedeulerlagrange}
    \begin{cases}
        \dot q_t &= p_t - \sqrt{2} \mathrm{Re} \langle \alpha_t, k B_{q_t}^\Lambda \rangle_2, \\
        \dot p_t &= \sqrt{2} \mathrm{Re} \langle \alpha_t, i k (f_{q_t}^K + (k.p_t) B_{q_t}^\Lambda) \rangle_2 + 2 \sum_{j=1}^d \mathrm{Re} \langle \alpha_t, k_j B_{q_t}^\Lambda \rangle_2 . \mathrm{Re} \langle \alpha_t, -i k_j  k B_{q_t}^\Lambda \rangle_2, \\
        i \dot \alpha_t &= \alpha_t + \frac{1}{\sqrt{2}} f_{q_t}^K -\frac{k.p_t}{\sqrt{2}} B_{q_t}^\Lambda + \sum_{j=1}^d k_j B_{q_t}^\Lambda \mathrm{Re} \langle \alpha_t, k_j B_{q_t}^\Lambda \rangle_2.
    \end{cases}
\end{equation}

To link these two PDEs, we will study the dressing more in depth in Section 2. To do so, we add a dynamical parameter $\theta \in \R$ and consider $U_\hbar^\Lambda (\theta) = (U_\hbar^\Lambda)^\theta$ as an evolution group. We will then study the semiclassical limit of the system when evolved with the generator of the group of isometries $U_\hbar^\Lambda (\theta)$, which is $\pi_\hbar(B_{\qh,\hbar}^\Lambda) := \phi_\hbar ( i B_{\qh,\hbar}^\Lambda)$. This leads to the definition of the dressing energy functional
\begin{align*}
    \mathscr E_{\text{D}}^\Lambda : \begin{array}{ll} \Hi &\longrightarrow \R \\
        (q,p,\alpha) &\longmapsto \sqrt{2} \mathrm{Re} \langle \alpha, i B_{q}^\Lambda \rangle_2.
    \end{array}
\end{align*}

The associated Euler-Lagrange equations are
\begin{equation}
    \label{dressingeulerlagrange} 
    \begin{cases}
        \dot q_\theta = 0 \\
        \dot p_\theta = -\sqrt{2} \mathrm{Re} \langle \alpha, k B_{q_\theta}^\Lambda \rangle_2  \\
        i \dot \alpha_\theta = \frac{i}{\sqrt{2}} B_{q_\theta}^\Lambda.
    \end{cases}
\end{equation}

As one can see, \eqref{undressedeulerlagrange} and \eqref{dressingeulerlagrange} look easier to study than \eqref{dressedeulerlagrange}. It is indeed the case, and we will prove in Section $2$ the following result :
\begin{thm} \label{clthm}
    The initial value problems \eqref{undressedeulerlagrange}, \eqref{dressedeulerlagrange} and \eqref{dressingeulerlagrange} are globally well posed in $\mathscr C^1(\R,\Hi)$. Furthermore the respective flows $\Phi(t)$, $\hat \Phi (t)$ and $\mathcal D (t)$
    are Borel measurable and are linked via the relation
    $$
    \hat \Phi(t) = \mathcal D(1) \circ \Phi(t) \circ \mathcal D(-1).
    $$
    Finally, for all $\theta \in \R$, $\mathcal D (\theta) $ is a symplectomorphism of $(\Hi,\sigma)$.
\end{thm}

\begin{rem}
    We say that the initial value problem \eqref{undressedeulerlagrange} (resp. \eqref{dressedeulerlagrange} and \eqref{dressingeulerlagrange}) are globally well-posed in $\mathscr C^1(\R,\Hi)$ if for all initial condition $u_0 \in \Hi$, there exists a unique solution $u \in \mathscr C^1(\R, \Hi)$ satisfying \eqref{undressedeulerlagrange} (resp. \eqref{dressedeulerlagrange} and \eqref{dressingeulerlagrange}) and $u(0) = u_0$. This ensures that, for instance, the flow of \eqref{undressedeulerlagrange}, defined as
    \begin{align*}
    \Phi(t) : \begin{array}{ll} \Hi &\longrightarrow \mathscr C^1 (\R,\Hi) \\
        u_0 &\longmapsto u(t),
    \end{array}
\end{align*}
where $u \in \mathscr C^1(\R, \Hi)$ is the solution of \eqref{undressedeulerlagrange} starting from $u(0) = u_0$, evaluated at time $t$, is well defined.
\end{rem}

In other words, the dynamics of the dressed equation, which looks more complicated to study than the others, is in fact well described by the undressed dynamics.

\subsection{Difficulty of removing the cutoff}
The Gross transform is of crucial use when looking at ground states of the Hamiltonian since it removes the divergence of the energy going at $-\infty$. In most cases the particle is kept quantum, but in our semiclassical setting, the usual techniques are not straightforward to use when one tries to remove the cutoff. As an example, the Lieb-Yamazaki trick, see Proposition \ref{lemmaA1}, has no physical meaning when considering a classical particle since the bound, while uniform in $\Lambda$, explodes in the limit $\hbar \rightarrow 0$:
$$\abs{ \langle \psi, H_\hbar^\infty \psi \rangle } \leq  \frac{1}{\hbar} ( \varepsilon \langle \psi, H_\hbar^0 \psi \rangle +  C_\varepsilon \norm{\psi}^2).$$
One could also try to define the limit by considering, whenever it makes sense, 
$$
\hat{\mathscr E}^\infty(q,p,\alpha) = \underset{\Lambda \rightarrow \infty} \lim  \hat{\mathscr E}^\Lambda(q,p,\alpha).
$$
The only ill-behaved term in the latter expression would be
$$\abs{p - \sqrt{2} \mathrm{Re} \langle \alpha , k B_{q}^\infty \rangle_2}.$$
Since $ k B_{q}^\infty \notin L^2$, we cannot consider $\alpha \in L^2$, otherwise the classical particle would only be defined almost everywhere, due to the lack of regularity. Therefore we could ask $\alpha$ to be more regular. Then the Hamilton equations would give in the undressed case
$$
i \dot \alpha_t = \alpha_t + \frac{1}{\sqrt{2}} f_{q_t}^\infty,
$$
and in the dressed case,
$$
i \dot \alpha_t = \alpha_t + \frac{1}{\sqrt{2}} f_{q_t}^K -\frac{k.p_t}{\sqrt{2}} B_{q_t}^\infty + \sum_{j=1}^d k_j B_{q_t}^\infty \mathrm{Re} \langle \alpha_t, k_j B_{q_t}^\infty \rangle_2.
$$
Since $f^\infty$ and $B^\infty \notin L^2$ are source term for these equations, we would expect for $\alpha_t$ to have at most the same regularity, making it unclear which space of functions to choose for the field component. Nonetheless we observe that the dressed energy $\hat{\mathscr E}^\infty(q,p,\alpha)$ remains bounded from below, even though it is defined only almost everywhere in $(q,p)$; though we are unable to define it properly, we think that it would give a stable evolution.
\\ \\
In a forthcoming article \cite{Gau25}, we consider using the same techniques the mean-field limit for the polaron, that is, $n$ quantum particles interacting with a bosonic field in the limit $n \hbar \rightarrow 1$, and we will prove that in this regime the classical system of PDEs make sense in the limit $\Lambda \rightarrow \infty$ and can be well studied and derivated using Wigner measure tools . 

\subsection{The quantum-classical transition} 
We define the set of \textit{density matrices} to be
$$
\mathfrak S^1 (\Gamma_s (\Hi))_{+,1} = \{ \rho \in \mathscr B (\Gamma_s (\Hi)) | \rho = \rho^*, \rho \geq 0, \Tr \rho = 1 \},
$$
where $\mathscr B (\Gamma_s (\Hi))$ denotes the set of bounded operators acting on $\Gamma_s (\Hi)$. The main result of this paper is the following :
\begin{thm} \label{wignertheorem}
    Let $(\rho_\hbar)_\hbar \subset \mathfrak S^1 (\Gamma_s (\Hi))_{+,1}$ be a family of density matrices such that there exists $C>0$ such that \begin{equation}
        \forall \hbar \in (0,1), \Tr(\rho_\hbar (\qh^2 + \ph^2 + N_\hbar)) \leq C.
    \end{equation} and let  $\hbar_n \rightarrow 0$ such that $\mathscr{M} (\rho_{\hbar_n} , n \in \N ) = \{ \mu \}$. Then,
    $$\begin{cases}
    \mathscr{M}(e^{it/{\hbar_n} \hat{H}_{\hbar_n}} \rho_{\hbar_n} e^{-it/{\hbar_n} \hat{H}_{\hbar_n}}, n \in \N ) = \{ \hat{\Phi}(t)_* \mu \}, \\
    \mathscr{M}(e^{it/{\hbar_n} H_{\hbar_n}} \rho_{\hbar_n} e^{-it/{\hbar_n} H_{\hbar_n}}, n \in \N) = \{ \Phi(t)_* \mu \}, \\
    {\mathscr{M}(U_{\hbar_n}^* \rho_{\hbar_n} U_{\hbar_n}, n \in \N) = \{ \mathcal D(-1)_* \mu \},
    }
    \end{cases}
    $$
    where the flows $\Phi(t), \hat{\Phi}(t), \mathcal D(-1)$ are the flows of \eqref{undressedeulerlagrange}, \eqref{dressedeulerlagrange} and \eqref{dressingeulerlagrange} respectively.
\end{thm}

Here the notation $f_* \mu$ stands for the push-forward measure of $\mu \in \mathscr P(\Hi)$ by the Borel measurable map $f : \Hi \rightarrow \Hi$, meaning that for all Borel sets $B$ of $\Hi$, $f_* \mu (B) = \mu (f^{-1}(B))$.

Let us first explain the strategy of the proof : 
\begin{itemize}
    \item \textit{Step 0.} Since 
    $$
    e^{it/\hbar \hat{H}_\hbar} \rho_\hbar e^{-it/\hbar \hat{H}_\hbar} = U_\hbar   e^{it/\hbar H_\hbar}  U_\hbar^* \rho_\hbar  U_\hbar   e^{-it/\hbar H_\hbar}  U_\hbar^*,  $$
    we only have to understand the semiclassical limit of dressing and of the undressed evolution to understand the dressed evolution. Considering a more general setting, we can group the two scaling limits under the same framework.
    \item \textit{Step 1.} We propagate in time the uniform estimate $\Tr(\rho_\hbar (\qh^2 + \ph^2 + N)) \leq C$ for dressed state as well as the undressed evolution to ensure existence of the Wigner measures.
    \item \textit{Step 2.} We establish the Duhamel formula in the interacting picture :
    $$
    \Tr(W_\hbar(\xi) \tilde{\rho_\hbar} (t) ) = \Tr( W_\hbar(\xi) \tilde{\rho_\hbar} ) - \frac{i}{\hbar} \int_0^t \Tr \left( [W_\hbar (\xi), H_I(s)] \tilde{\rho}_\hbar (s) \right) \di s
    $$
    and expand the commutator $[W_\hbar (\xi), H_I(s)] = \hbar M(s,\xi,\hbar) + \hbar^2 R(s,\xi,\hbar)$ with $M$ and $R$ being $H^0_{\hbar}$-bounded operators. The latter allow us to prove uniform continuity with respect to $\hbar$ of the characteristic functionals, and find a common subsequence to take the limit.
    \item \textit{Step 3.} We recognize the form of the Wigner measures using the uniqueness of the initial value problems on the classical side.
\end{itemize}

\section{Study of the classical model}

We first study the classical polaron dynamics. We recall the definition of the energy functional
$$
\mathscr E^\Lambda (q,p,\alpha) =   \frac{\abs{p}^2}{2} + \norm{\alpha}_2^2 + \sqrt{2} \mathrm{Re} \langle \alpha, f^\Lambda_q \rangle_2.
$$
The Hamilton equations are :
\begin{equation} \label{hamiltonjacobi}
\left\{
    \begin{array}{lll}
        \partial_t q_t = p_t \\
        \partial_t p_t = -\sqrt{2} \mathrm{Re} \langle \alpha_t , -i k f^\Lambda_q \rangle \\
        i \partial_t \alpha_t = \alpha_t + \frac{1}{\sqrt{2}} f^\Lambda_q
    \end{array}.
\right.
\end{equation}

\subsection{Global well-posedness of the undressed model}
The energy functional can be estimated by
\begin{equation} \label{energybyhalfnorm}
\mathscr E_{g,F} (q,p,\alpha) \lesssim_\Lambda 1 + \abs{p}^2 + \norm{\alpha}^2.
\end{equation}

Set $u = (q,p,\alpha) \in \Hi$ and rewrite \eqref{hamiltonjacobi} as the initial value problem
\begin{equation}
    \label{ivp}
\frac{\di u(t)}{\di t} = (\mathscr L  + \mathscr N) u (t) 
\end{equation}
with $\mathscr L u := ( p , 0 , -i  \alpha)$ the linearity and $\mathscr N u := (0,- \sqrt{2} \mathrm{Re} \langle \alpha , -i k  f_q^\Lambda \rangle_2, -i \frac{1}{\sqrt{2}} f_q^\Lambda )$ the nonlinearity. The free flow is
\begin{equation}
    \label{freeflow}
\Phi^0_t (q,p,\alpha) = (q + t p , p, e^{- i t} \alpha), 
\end{equation}
We consider the interaction picture, i.e. the initial value problem
\begin{equation}
    \label{ivplambda}
\frac{\di u(t)}{\di t} = X(t,u(t)), 
\end{equation}
and the vector field $X : \R \times \Hi \rightarrow \R \times \Hi$ is defined as
\begin{equation}
\label{interactingflow}
\begin{aligned}
X(t,(q,p,\alpha)) :=& \Phi^0_{-t} \circ \mathscr N  \circ  \Phi^0_t (q,p,\alpha) \\
=& \begin{pmatrix} \sqrt{2} t \mathrm{Re} \langle e^{-i t } \alpha , - i k f_{q+ t p} \rangle_2 \\   -\sqrt{2} \mathrm{Re} \langle e^{-i t} \alpha, - i k f_{q+ t p} \rangle_2  \\  -i e^{i t} \frac{1}{\sqrt{2}} f_{q+t p} 
\end{pmatrix}.
\end{aligned}
\end{equation}

\begin{lem}
The vector fields $\mathscr L, \mathscr N , X : \Hi \rightarrow \Hi$ and $X : \R \times \Hi \rightarrow \Hi$ are continuous.    
\end{lem}

\begin{proof}
For $u_1 = (q_1,p_1,\alpha_1),u_2 = (q_2,p_2,\alpha_2) \in \Hi$, a straight computation gives
\begin{align*}
\norm{\mathscr L(u_1)- \mathscr L(u_2)} &\leq_g \norm{u_1-u_2},
\\
\norm{\mathscr N(u_1)- \mathscr N(u_2)} &\lesssim_{g,F} (1+\norm{\alpha_1}_2+\norm{\alpha_2}_2) \norm{u_1-u_2}, \\
\norm{X(t, u_1)- X(t, u_2)} &\lesssim_{g,F} (1+\abs{t}) (1+\norm{\alpha_1} + \norm{\alpha_2}) \norm{u_1-u_2}. \end{align*}

\end{proof}
The equivalence between \eqref{ivplambda} and \eqref{ivp} is given by

\begin{prop} \label{solution} Let $I$ be an interval of $\R$ such that $0 \in \mathring I$. The following statements are equivalent :
\begin{enumerate}
    \item $u \in \mathscr C^1(I,\Hi)$ satisfies \eqref{ivplambda}.
    \item $u \in \mathscr C(I,\Hi)$ satisfies the Duhamel formula : for all $t \in I$,
    $$
    u(t) = u(0) + \int_0^t X(s,u(s)) \di s.
    $$
    \item $t \mapsto \Phi_t^0 ( u(t) )$ belongs to $\mathscr C^1(I,\Hi)$ and satisfies \eqref{ivp} .
\end{enumerate}
\end{prop}

\begin{defn}
    A map $u \in \mathscr C^1(I,\Hi)$ satisfying the equivalent statements of Proposition \ref{solution} is called a \textit{strong solution} of \eqref{ivp} or \eqref{ivplambda}.
\end{defn}

\begin{proof}[Proof of Proposition \ref{solution}]
The equivalence between \textit{1.} and \textit{2.} is a consequence of the continuity of $X$. The equivalence between \textit{1.} and \textit{3.} is a straightforward computation: if $u \in \mathscr C^1(I,\Hi)$ satisfies \eqref{ivplambda}, then $v : t \mapsto \Phi_t^0 ( u(t) ) \in \mathscr C^1(I,\Hi)$ by composition and 
\begin{align*}
\frac{\di v}{\di t} &= \mathscr L v(t) + \Phi_t^0 (X(t,u(t))) \\
&= \mathscr L v(t) + \mathscr N \circ \Phi_t^0 (u(t)) \\
&= (\mathscr L + \mathscr N) v(t).
\end{align*}
The converse implication is similar.
\end{proof}
We first investigate global-well posedness for the initial value problem $\eqref{ivp}$. By Proposition \ref{solution} the global-well posedness in the interaction picture will then follow.
\begin{prop} (uniqueness) \label{uniqueness} Let $I$ be an interval such that $0 \in \stackrel{\circ}{I}$. Let $u_1,u_2 \in \mathscr C(I,\Hi) $ be solutions to $(\ref{ivp})$ such that $u_1 (0) = u_2 (0)$. Then $u_1 \equiv u_2$.
\end{prop}

\begin{proof} Let $t \in I$. 
    \begin{align*}
        \norm{u_1(t) - u_2(t)} \lesssim \int_0^t \norm{u_1(s)-u_2(s)} (1+ \norm{u_1(s)}+\norm{u_2(s)} ) \di s,
    \end{align*}
    and Gronwall's lemma yields the result.
\end{proof}

We now prove the local existence of strong solutions of \eqref{ivp} by means of a standard fixed point theorem.

\begin{prop} (local existence) Let $R>0$. There exists $T(R)>0$ such that for all $u_0 \in \Hi$ with $\norm{u_0}\leq R$, there exists $u \in \mathscr C([-T(R),T(R)],\Hi)$ solution to  $\eqref{ivp}$ with initial condition $u (0) = u_0$.
\end{prop}

\begin{proof} Let $T,M > 0$ to be fixed later. Let $I = [-T,T]$. We will use Banach's fixed point theorem. We set 
$$
E = \{ u \in \mathscr C(I, \Hi ) \text{ such that } \norm{u}_\infty = \underset{t \in I}{\sup} \norm{u(t)} \leq M \},
$$
equipped with the distance 
$$
d (u_1,u_2) := \norm{u_1-u_2}_\infty,
$$
which makes $(E,d)$ complete. Fix an initial condition $u_0 \in  \Hi$.
Let $u \in E$ and define
$$ 
\mathscr{K} u (t) = u_0 + \int_0^t (\mathscr L + \mathscr N)(s,u(s)) \di s.
$$
First, let's see that $\mathscr{K} u \in E$. Set $M = 2 \norm{u_0}$. There exists $C>0$ such that
\begin{align*}
 \norm{\mathscr{K} u(t)} &\leq \norm{u_0} +  C \int_0^t  \norm{u(s)} (1+ \norm{u(s)} ) \di s \\
 &\leq R + C T M (1+ M).
\end{align*}
We set $M = 2 R$, and for $$T \leq \frac{1}{2 C (1+ M)},$$
$\norm{\mathscr{K} u(t)} \leq R$ so that $\mathscr K$ maps $E$ to itself.
Now, check that for any $u_1, u_2 \in E$,
\begin{align*}
 \norm{\mathscr{K} u_1(t) - \mathscr{K} u_2(t)} &\leq  C \int_0^t \norm{u_1(s)-u_2(s)} (1+\norm{u_1(s)}+\norm{u_2(s)} ) \di s  \\
 & \leq C T \norm{u_1 - u_2}_\infty (1 + 2 M),
\end{align*}
so that 
$$ 
\norm{\mathscr{K} u_1-\mathscr{K} u_2}_\infty \leq \frac{1}{2} \norm{u_1-u_2}_\infty
$$
for 
$$T \leq  \frac{1}{2 C (1 + 2 M)}.$$
By Banach's fixed point theorem we have local existence up to time $$ T = T(R) := \frac{1}{2 C (1 + 4 R)}.$$
\end{proof}
Since the time of existence only depend on the norm of the initial data, it will follow from a blowup criterion that if the maximal time of existence is finite then the solution blows up in finite time. In the next Lemma we show that this conclusion cannot hold.
\begin{lem} \label{boundqpalpha}
    Let $u \in \mathscr C([-T,T],\Hi)$ be a solution to  $(\ref{ivp})$. Then there exists $C>0$ such that for all $t \in [-T,T]$, $\norm{u(t)} \leq C e^{C \abs{t}}$.
\end{lem}

\begin{proof}
    The solution preserves the energy functional: for all $t \in [-T,T]$, $\mathscr E_{g,F}(u(t)) = \mathscr E_{g,F}(u(0)).$ Using this and $\eqref{energybyhalfnorm}$, one deduces that the maps $t \mapsto \abs{p(t)}^2$ and $t \mapsto \norm{\alpha(t)}$ are bounded uniformely with respect to $T$. For the remaining term, one checks that 
    \begin{align*}
        \frac{\di (1+\abs{q(t)}^2)}{\di t} &= 2 p(t) . q(t) \\
        & \leq 2 \abs{p(t)} \abs{q(t)} \\
        & \leq C \abs{q(t)} \\
        & \leq C (1+\abs{q(t)}^2).
    \end{align*}
    Hence by Gronwall's lemma, $\abs{q(t)} \leq 1+ \abs{q(t)}^2 \leq C e^{C t}$, which yields the result.
\end{proof}

\begin{prop} (global existence) \label{globalexistence}
Let $u_0 \in \Hi$. There exists $u \in \mathscr C(\R,\Hi)$ a solution to $(\ref{ivp})$ such that $u (0) = u_0$.
\end{prop}

\begin{proof}
    Without loss of generality, we only prove global existence for positive times. Consider a maximal solution $u \in \mathscr C([0,T_{\mathrm{max}}),\Hi)$ of $(\ref{ivp})$, starting from $u_0 \in \Hi$. For a contradiction, assume that $T_{\mathrm{max}} < \infty.$ By Lemma \ref{boundqpalpha}, $\underset{t \in [0, T_{\mathrm{max}})}{\sup} \norm{u(t)}\leq C e^{C T_{\mathrm{max}}} =: R <\infty.$ Consider a sequence of times $t_k \underset{k \rightarrow \infty}{\longrightarrow} T_{\mathrm{max}}.$ From the proof of local existence, we know that the time of existence only depends on the norm of the initial data. Let $k$ such that $t_k + T(R) > T_{\mathrm{max}}.$ We thus extend $u$ by taking the solution of $(\ref{ivp})$ with initial conditions $\norm{u(t_k)}$. By uniqueness, $u$ is a well defined solution on $[0,t_k + T(R))$, which contradicts the maximality of $T_{\mathrm{max}}$. We conclude that $T_{\mathrm{max}} = \infty$ and that the solution is global in time.
\end{proof}

\begin{cor} (global-well posedness in the interacting picture) Let $u_0 \in \Hi$. There exists a unique $u \in \mathscr C(\R,\Hi)$ solution to $(\ref{ivplambda})$ such that $u (0) = u_0$.
\end{cor}

\begin{proof}
    The existence is provided by Proposition \ref{solution} and \ref{globalexistence}. The uniqueness is proved in the same way as Proposition \ref{uniqueness}.
    
\end{proof}

\subsection{Propagation of regularity}
 The equation for the field
    $$
    i \partial_t \alpha_t =  \alpha_t + \frac{1}{\sqrt{2}} f^\Lambda_q
    $$
    yields the following Duhamel formula for $\alpha_t$:
    $$
    \alpha_t = e^{-i t} \left( \alpha_0 - \frac{i}{\sqrt{2}}  \int_0^t e^{i s} f^\Lambda_{q_s} \di s\right).
    $$
Having an integral representation of $\alpha_t$ allow us to prove that the flow also preserves functional spaces with more regularity in the field component.
\begin{defn}
Let $s \in \R$. Define the space
$$
\Hi^s := \R^{2 d} \oplus \mathcal F H^s (\R^d),  
$$
where 
$$
\mathcal F H^s(\R^d) = \{ \alpha \in L^2 (\R^d), \mathcal F^{-1}(\alpha) \in H^s \} = \{\alpha \in L^2 (\R^d), \langle k \rangle^s \alpha(k) \in L^2 (\R^d) \},
$$
and $\langle k \rangle^2 = 1+\abs{k}^2$ denotes the Japanese bracket.
\end{defn}

\begin{prop}(Propagation of regularity) \\
Let $s \geq 0$ and $u_0 \in \Hi^s$. There exists a unique $u \in \mathscr C(\R,\Hi^s)$ solution to $(\ref{ivplambda})$ such that $u (0) = u_0$.
\end{prop}

\begin{proof}
    Consider the solution $u_t = (q_t,p_t, \alpha_t) \in \mathscr C (\R, \Hi)$ to $(\ref{ivplambda})$ given by Proposition \ref{globalexistence} with initial data $(q_0,p_0,\alpha_0) \in \Hi^s \subset \Hi$.
    Thus,
    \begin{align*}
        \langle k \rangle^s  \alpha_t (k) = e^{- i t} \left( \langle k \rangle^s \alpha_0 (k) - \frac{i}{\sqrt{2}}  \int_0^t e^{i s} \langle k \rangle^s \frac{e^{-i k . q_s} \Id_{\abs{k} \leq \Lambda}}{\abs{k}^{\frac{d-1}{2}}}f^\Lambda_{q_s} \di s \right) \in L^2 (\R^d , \di k).
    \end{align*}
    Hence, we deduce that $\alpha_t \in \mathcal F H^s$ and that $u \in \mathscr C (\R, \Hi^s)$. This proves existence. Uniqueness is a consequence of that of $\Hi$, since $\Hi^s \subset \Hi$.
\end{proof}

\subsection{Properties of the classical dressing}
We now investigate in more detail the dressing and its symplectic properties. Recall the expression for the equations of the dressing:
\begin{equation}
    \begin{cases}
        \dot q_\theta = 0 \\
        \dot p_\theta = -\sqrt{2} \mathrm{Re} \langle \alpha, k B_{q_\theta}^\Lambda \rangle_2  \\
        i \dot \alpha_\theta = \frac{i}{\sqrt{2}} B_{q_\theta}^\Lambda.
    \end{cases}
\end{equation}
The solutions are explicit:
$$
\begin{pmatrix}
q_\theta \\ p_\theta \\ \alpha_\theta 
\end{pmatrix}
=
\begin{pmatrix} 
    q_0 \\ p_0 - \sqrt{2} \theta \mathrm{Re} \langle \alpha_0 , k  B^\Lambda_{q_0} \rangle_2   \\ \alpha_0 +  \frac{\theta}{\sqrt{2}} B^\Lambda_{q_0}
\end{pmatrix}
=
\mathcal D(\theta)
\begin{pmatrix}
    q_0 \\ p_0 \\ \alpha_0
\end{pmatrix},
$$
for an initial condition $(q_0,p_0,\alpha_0) \in \Hi$.

\begin{prop}
    For all $\theta \in \R$, $\mathcal D(\theta)$ is a symplectomorphism of $(\Hi, \sigma)$, i.e. it is differentiable at every $u_0 \in \Hi$ and for all $u_1,u_2 \in \Hi$ we have
    $$
    \sigma( \di \mathcal D(\theta)_{u_0}(u_1), \di \mathcal D(\theta)_{u_0} (u_2) ) = \sigma( u_1 , u_2 ).
    $$
\end{prop}

\begin{proof}
    It is an explicit computation. First set $u_i = (q_i,p_i,\alpha_i)$ for $i \in \{0,1,2 \}$. It is clear that $\mathcal D(\theta)$ is differentiable at $u_0$ and its differential is
    $$
    \di \mathcal D(\theta)_{u_0} (u_1) = 
    \begin{pmatrix} 
    q_1 \\ p_1 - \sqrt{2} \theta \mathrm{Re} \langle \alpha_1 ,  k B^\Lambda_{q_0} \rangle  - \sqrt{2} \theta \mathrm{Re} \langle \alpha_0, (-i k . q_1)  k B^\Lambda_{q_0} \rangle \\ \alpha_1 - \frac{\theta}{\sqrt{2}} i  k . q_1 B^\Lambda_{q_0}
    \end{pmatrix}.
    $$
    Hence, using that $ \mathrm{Im} \langle -i k . q_1 B^\Lambda_{q_0}, -i  k . q_2 B^\Lambda_{q_0} \rangle = \mathrm{Im} \langle k . q_1 B^\Lambda,  k . q_2 B^\Lambda \rangle = 0$, and simplyfing the terms with the same underlining we obtain:
    \begin{align*}
        \sigma(\di \mathcal D(\theta)_{u_0} (u_1),\di \mathcal D(\theta)_{u_0} (u_2)) =& q_1 . (p_2 - \dashuline{\sqrt{2} \theta \mathrm{Re} \langle \alpha_2 ,  k B^\Lambda_{q_0} \rangle}  - \underline{\sqrt{2}  \theta \mathrm{Re} \langle \alpha_0, (-i k . q_2)   k B^\Lambda_{q_0} \rangle}) \\
        -& q_2 . (p_1 - \dotuline{\sqrt{2} \theta \mathrm{Re} \langle \alpha_1 ,  k B^\Lambda_{q_0} \rangle}  - \underline{\sqrt{2} \theta  \mathrm{Re} \langle \alpha_0, (-i k . q_1)  k B^\Lambda_{q_0} \rangle}) \\
        +& 2 \mathrm{Im}\langle \alpha_1,\alpha_2 \rangle \\ +& \dotuline{2 \mathrm{Im}\langle \alpha_1, -  \frac{\theta}{\sqrt{2}} i k.q_2 B^\Lambda_{q_0} \rangle} + \dashuline{2 \mathrm{Im}\langle  -  \frac{\theta}{\sqrt{2}} i k.q_1 B^\Lambda_{q_0}, \alpha_2  \rangle} \\
        =& q_1.p_2 - q_2.p_1 + 2 \text{Im} \langle \alpha_1, \alpha_2 \rangle \\
        =& \sigma(u_1,u_2).
    \end{align*}
\end{proof}

As for the polaron dynamics, the regularity of the initial condition is propagated along the flow of the dressing.

\begin{prop}(Propagation of regularity)
    For all $\theta \in \R$, $\mathcal D(\theta) (\Hi^s) \subset \Hi^s$.
\end{prop}

\begin{proof}
    Let $(q_0,p_0,\alpha_0) \in \Hi^s$. Then,
    $$
    \alpha_\theta = \alpha_0 +  \frac{\theta}{\sqrt{2}} B_{q_0}^\Lambda.
    $$
    Hence,
    $$
    \langle k \rangle^s \alpha_\theta (k) = \langle k \rangle^s \alpha_0 (k) -  \frac{\theta}{\sqrt{2}} \langle k \rangle^s  \frac{e^{-i k . q}}{\abs{k}^{\frac{d-1}{2}}} \Id_{K \leq \abs{k} \leq \Lambda} \in L^2 (\R^d, \di k).
    $$
\end{proof}

\subsection{Link between the undressed and dressed system}

We are now able to link the two descriptions of the system, dressed and undressed. We recall that the dressed system is given by the energy functional
$$
    \hat{\mathscr E}^\Lambda (q,p,\alpha) = \frac{\abs{p}^2 }{2} + \norm{\alpha}_2^2 + \sqrt{2} \mathrm{Re} \langle \alpha, f_q^K \rangle_2 - \sqrt{2}  \mathrm{Re} \langle \alpha, (k.p) B_q^\Lambda \rangle_2 + (\mathrm{Re} \langle \alpha, k B_q^\Lambda \rangle_2)^2,
$$
or, in a factorized way,
$$
\hat{\mathscr E}^\Lambda (q,p,\alpha)  = \frac{\abs{p - \sqrt{2} \mathrm{Re} \langle \alpha , k B_q^\Lambda \rangle_2}^2}{2} + \norm{\alpha + \frac{1}{\sqrt{2}}f_q^K}_2^2 - \frac{\norm{f^K}^2_2}{2}.
$$
\begin{prop}For every $t \in \R$, we have $ \hat{\mathscr E}^\Lambda = \mathscr E^\Lambda \circ \mathcal D(1), $ or, in terms of flows,$$
\hat \Phi (t) = \mathcal D(-1) \circ \Phi (t) \circ \mathcal D(1).
$$
\end{prop}

\begin{proof}
Noticing that $f_q^\Lambda + B_q^\Lambda = f_q^K$, we have that for all $(q,p,\alpha) \in \Hi$,
    \begin{align*}
        \mathscr E^\Lambda \circ \mathcal D(1) (q,p,\alpha) &= \mathscr E^\Lambda (p-\sqrt 2 \mathrm{Re} \langle \alpha, k B_q \rangle ,q,\alpha + \frac{B_q}{\sqrt{2}}) \\
        &=  \frac{\abs{p - \sqrt{2} \mathrm{Re} \langle \alpha , k B_q^\Lambda \rangle_2}^2}{2} + \norm{\alpha + \frac{f_q^\Lambda +B_q^\Lambda}{\sqrt{2}}}_2^2 - \frac{\norm{f^\Lambda}^2_2}{2} \\
        &= \hat{\mathscr E}^\Lambda (q,p,\alpha). 
    \end{align*}
\end{proof}

\section{Semiclassical limit of the quantum models}

\subsection{Tools of second quantization}
In this subsection we briefly review the tools of second quantization that will be needed in the following. Let us recall our Hilbert space
$$
\Gamma_s (\Hi) =  \Gamma_s (\C^d \oplus \Hi_f) \cong L^2 (\R^d) \otimes  \Gamma_s (\Hi_f).
$$
The vacuum is $\Omega = \varphi_0 \otimes \Omega_f$, where $\varphi_0 (x) = (\pi \hbar)^{- \frac{1}{4} } e^{- \frac{ \abs{x}^2 }{2 \hbar}  }$ and $\Omega_f = 1 \oplus 0 \oplus 0 \oplus...$ is the field vacuum. For any self-adjoint operator $(A,D(A))$ acting on $\Hi_f$, we define its second quantization as the closure of the operator defined on $\bigoplus_{n \in \N}^{\text{alg}} D(A)^{\otimes_s n}$ by
$$
\di \Gamma_\hbar (A) |_{D(A)^{\otimes_s n}} = \hbar \sum_{i=1}^n \Id \otimes ... \otimes \underbrace{A}_{i^{\text{th}} \text{position}} \otimes ... \otimes \Id.
$$
The closure, still denoted $\di \Gamma_\hbar (A)$ is self-adjoint and its action on the field vacuum is $\di \Gamma_\hbar (A) \Omega_f = 0$. In the particular case $A = \Id_{L^2 (\R^d)}$, we call $\di \Gamma_\hbar (h)$ the \textit{number operator}, or \textit{number of particles}, and denote it $N_\hbar$. All of these operators are naturally extended to self-adjoints operators on the full Hilbert space.
\\
\\
For any $\xi \in \Hi_f$, the creation and annihilation operators $a_\hbar^* (\xi)$, $a_\hbar (\xi)$ are defined as follows: for $\Psi = (\Psi^{(n)})_{n \in \N} \in \bigoplus_{n \in \N}^{\text{alg}} \Hi_f^{\otimes_s n}$, let
\begin{align*}
&(a_\hbar^* (\xi) \Psi)^{0} := 0, \\
&(a_\hbar^* (\xi) \Psi)^{(n+1)} := \sqrt{\hbar (n+1)} S_{n+1} ((\xi \otimes \Id_{\Hi_f^{\otimes_s n}} ) \Psi^{(n)}) , \ n \in \N.
\end{align*}
Here $S_n$ denotes for any $n \in \N$ the orthogonal projection of $ \Hi_f^{\otimes n}$ onto its symmetric subspace $  \Hi_f^{\otimes_s n}$. This defines a closable operator, that we close while keeping the same notation, and which is named name creation operator. The annihilation operator is then defined as the adjoint
$$
    a_\hbar (\xi) = (a_\hbar^* (\xi))^*.
$$
Similarly, on the particle side, we have the position and momentum operators $\qh_j$ and $\ph_j$, $j \in \{1,...,n\}$.
These operators satisfy the so-called \textit{canonical commutation relations}: for all $\alpha,\beta \in \Hi$,
\begin{equation}
    \label{CCR} 
    \begin{cases}
    [a_\hbar(\alpha),a_\hbar^* (\beta)] &= \hbar \langle \alpha, \beta \rangle \Id,      \\
    [a_\hbar(\alpha), a_\hbar(\beta)]&=0,   \\
    [a^*_\hbar(\alpha), a^*_\hbar(\beta)]&=0,
    \end{cases}
    \ \ \ \ \
    \begin{cases}
        [\qh_i , \ph_j] &= i \hbar \delta_{i,j}, \\
        [\qh_i, \qh_j]& = 0, \\
        [\ph_i, \ph_j] &= 0.
    \end{cases}
\end{equation}
To avoid considerations on the domain, one often consider the \textit{Weyl operators} : for $\alpha \in \Hi_f$ and  $(q,p) = z \in \C^d$, define the particle and field Weyl operators
\begin{equation}
\label{Weyl}  
\begin{aligned}
W_\hbar(\alpha) &:= e^{\frac{i}{\sqrt{2}} (a_\hbar^* (\alpha) + a_\hbar (\alpha)  )} = e^{i \phi_\hbar (\alpha)},\\
T_\hbar (q,p) &:= e^{ i (p . \qh - q . \ph) } = e^{i \mathrm{Im}\langle \qh + i \ph , z \rangle }.
\end{aligned}
\end{equation}
$\phi_\hbar(\alpha) := \frac{1}{\sqrt{2}} (a_\hbar^* (\alpha) + a_\hbar (\alpha)  )$ is called the field operator, and we also define $\pi_\hbar(\alpha) := \phi_\hbar( i \alpha) = \frac{1}{\sqrt{2}} (a_\hbar^* (\alpha) - a_\hbar(\alpha))$. Rewritting \eqref{CCR}  in terms of Weyl operators gives the \textit{Weyl commutation relations}: for all  $\alpha, \beta \in \Hi_f$ and $z,w \in \C^d$,
\begin{equation}
    \label{WCR} 
    \begin{aligned}
    W_\hbar (\alpha) W_\hbar(\beta) &= e^{-i \frac{\hbar}{2} \mathrm{Im} \langle \alpha , \beta \rangle_2 }   W_\hbar (\alpha + \beta), \\
    T_\hbar (z) T_\hbar (w) &= e^{-i \frac{\hbar}{2} \text{Im} \langle z,w \rangle } T_\hbar (z+w).
    \end{aligned}
\end{equation}
Now, let us recall standard estimates on the creation and annihilation operators :
\begin{prop}
    Let $\xi \in \Hi_f$ be a bounded linear operator. Then $D(\sqrt{N_\hbar}) \subset D(a_\hbar^{\#}(\xi))$ and for all $\Psi \in D(\sqrt{N_\hbar})$,
    \begin{equation}
        \label{estimatecreation}
    \norm{a_\hbar^{\#} (\xi) \Psi } \leq \norm{\xi}_{\Hi_f} \norm{\sqrt{N_\hbar +1} \Psi}.
      \end{equation}
\end{prop}

\subsection{Wigner measures}

\begin{defn}
We say that $\mu \in \mathscr P (\Hi)$ is a Wigner measure of the family of density matrices $(\rho_\hbar)_\hbar \subset \mathfrak S^1 (\Gamma_s (\Hi))_{+,1}$ if there exists a subsequence $\hbar_n \rightarrow 0$ such that
$$
\underset{n \rightarrow \infty}{\lim} \Tr \left( \rho_{\hbar_n} T_{\hbar_n} \left( \frac{z}{2  i \pi} \right) \otimes W_{\hbar_n} \left( \frac{\alpha}{\sqrt{2} \pi} \right)    \right) = \int_{\Hi} e^{2 i \pi \mathrm{Re} \langle (z,\alpha), (w,\beta) \rangle} \di \mu (w,\beta).
$$
We denote the set of Wigner measures of $(\rho_{\hbar_n})_{\hbar_n}$ by $\mathscr{M} (\rho_{\hbar_n}, n \in \N)$. The right-hand side is in fact the Fourier transform of the probability measure $\mu$, and by Bochner's theorem, this entirely determines $\mu$.
\end{defn}
The following theorem, due to \cite{ammari2008ahp}, gives a sufficient condition for the existence of Wigner measures for a given family of density matrices:

\begin{prop} \label{wignerzied}
    Let $(\rho_\hbar)_\hbar \subset \mathfrak S^1 (\Gamma_s (\Hi))_{+,1}$ be a family of density matrices satisfying the uniform estimate:
    \begin{equation} \label{assumption} 
    \exists C>0, \forall \hbar \in (0,1), \Tr \left( \rho_\hbar (N_\hbar + \qh^2 + \ph^2) \right) \leq C.
    \end{equation}
    Then,
    $$
    \mathscr{M} (\rho_\hbar, \hbar \in (0,1)) \neq \emptyset.
    $$
    Moreover, 
    \begin{equation} \label{uniformmeasurebound}
    \forall \mu \in \mathscr{M} (\rho_\hbar, \hbar \in (0,1)), \int_{\Hi} \norm{(z,\alpha)}^{2} \di \mu(z,\alpha) < \infty.
    \end{equation}
\end{prop}

\subsection{Propagation of regularity estimates}

We consider a slightly more general setting : for $g \geq 0$, and a form factor $F \in L^\infty_q L^2(\R^d, (1+\abs{k}^2) \di k) $, define the following Hamiltonian:
$$
H_{\hbar,g,F} := g H_{\hbar}^0 + \phi_\hbar (F_{\qh}).
$$
We can treat under the same framework the undressed Hamiltonian ($(g,F) = (1, f^\Lambda)$) and the dressing ($(g,F) = (0, i B^\Lambda)$). The unitary evolution associated to $H_{\hbar,g,F}$ is denoted
$$
U_{\hbar,g,F}(t) := e^{-i \frac{t}{\hbar} H_{\hbar,g,F}}.
$$
The evolution of a density matrix $\rho_\hbar$ is then 
$$
\rho_\hbar (t) := U_{\hbar,g,F}(t)^* \rho_\hbar U_{\hbar,g,F} (t) = e^{i \frac{t}{\hbar} H_{\hbar,g,F}} \rho_\hbar e^{-i \frac{t}{\hbar}H_{\hbar,g,F}},
$$
and in the interaction picture:
$$
\tilde{\rho}_\hbar (t) = U_{\hbar,g,0}(t) \rho_\hbar (t) U_{\hbar,g,0} (t)^* = e^{-i g \frac{t}{\hbar}H_{\hbar}^0} e^{i \frac{t}{\hbar} H_{\hbar,g,F}} \rho_\hbar e^{-i \frac{t}{\hbar}H_{\hbar,g,F}} e^{i g \frac{t}{\hbar}H_{\hbar}^0}.
$$

\begin{prop}
    $H_{\hbar,g,F}$ is a self-adjoint operator. For $g>0$, its domain is $D(H_\hbar^0)$.
\end{prop}

\begin{proof}
For $g=0$, it is a consequence of the self-adjointness of the field operator. For $g>0$, we will use Rellich-Kato's theorem : for any $\Psi \in D(H_\hbar^0)$,
    \begin{align*}
        \norm{\phi_\hbar(F_{\qh}) \psi}^2 &\leq 2 \norm{F_{\qh}}^2 \norm{\sqrt{N_\hbar + 1} \Psi}^2 \\
        & \leq 2 \norm{F}_2^2  \langle \Psi , (N_\hbar + 1 ) \Psi \rangle^2 \\
        & \leq \eps \norm{H_\hbar^0 \Psi}^2 + C_\eps \norm{\Psi}^2,
    \end{align*}
where we used Cauchy-Schwarz inequality and $ a b \leq \eps a^2 + C_\eps b^2$ to get the last estimate.
\end{proof}

\begin{rem}
    If $F = f^\Lambda$ or $i B^\Lambda$, the r.h.s of the inequality grows like a power of $\Lambda$. This shows that this type of proof fails in the limit $\Lambda \rightarrow \infty$.
\end{rem}

\begin{prop}
    Let $(\rho_\hbar)_\hbar \subset \mathfrak S^1 (\Gamma_s (\Hi))_{+,1}$ be a family of density matrices satisfying \eqref{assumption}. Then, there exists $C>0$ such that 
    \begin{equation} \label{propreg} 
    \forall t \in \R, \forall \hbar \in (0,1), \
    \Tr (\rho_\hbar(t) (\qh^2 + \ph^2 + N_\hbar)) \leq C e^{\abs{t} C}, 
    \end{equation}
\end{prop}
\begin{proof}
Since $[H_\hbar, \qh] = 0$, $\Tr(\rho_\hbar (t) \qh^2) = \Tr(\rho_\hbar \qh^2) \leq C$, so we only have to estimate $\Tr(\rho_\hbar(t) H_\hbar^0)$. $\rho_\hbar$ is a density matrix, so we can diagonalize
    $$ \rho_\hbar = \sum_m \lambda_\hbar (m) \pure{\psi_\hbar (m)}, $$
    where the $\psi_\hbar (m)$ are of norm $1$, the $\lambda_\hbar(m)$ are nonnegative and $\sum_{m \in \N} \lambda_\hbar (m) = 1$. 
    Thus,
    $$ \rho_\hbar (t)  = \sum_m \lambda_\hbar (m) \pure{e^{-i t /\hbar H_\hbar} \psi_\hbar (m)}. $$
    Regarding \eqref{assumption}, we obtain that every $\psi_\hbar (m) \in D((H_\hbar^0)^{1/2})$. For all $\varphi \in D((H_\hbar^0)^{1/2})$,
    $\varphi(t) := e^{-i t /\hbar H_{\hbar,g,F}} \varphi \in D((H_\hbar^0)^{1/2})$. For $g>0$, it follows from the fact that $D(H_\hbar^0) = D(H_{\hbar,g,F})$. For $g=0$, it follows from the inclusion $D((H_{\hbar}^0)^{1/2}) \subset D(\phi_\hbar (F_{\qh}))$ that comes from \eqref{estimatecreation}. Furthermore, if $\norm{\varphi}=1$, then $\norm{\varphi(t)}=1$.
    
    Hence
    $ t \mapsto \langle \varphi(t), H_\hbar^0 \varphi(t) \rangle$ is differentiable and we have :
    \begin{align*}
        \frac{\di}{\di t} \langle \varphi (t), H_\hbar^0 \varphi (t) \rangle &= \frac{i}{\hbar} \langle \varphi (t), [H_{\hbar,g,F}, H_\hbar^0] \varphi (t) \rangle \\ 
        &= \frac{i}{\hbar} \langle \varphi (t), [H_{\hbar,g,F} -g  H_\hbar^0, H_\hbar^0] \varphi (t) \rangle.
    \end{align*}
    We now have to estimate the commutator $[H_{\hbar,g,F} - g  H_\hbar^0, H_\hbar^0] = [ \phi(F_{\qh} )  , \ph^2/2 + N_\hbar]$. We have 
      $$  [ \phi_\hbar(F_{\qh} )  , \ph^2/2 + N_\hbar] = \frac{\hbar}{i} \phi(i F_{\qh} ) + \frac{\hbar}{2} ( \phi(k . \ph F_{\qh}) + \phi( F_{\qh} k . \ph  ) ). $$
    Using the estimate \eqref{estimatecreation} and Cauchy-Schwarz inequality we get :
    \begin{align*}
    \abs{\frac{\di}{\di t} \langle \varphi (t), H_\hbar^0 \varphi (t)\rangle} &\lesssim \norm{F_{\qh}} \norm{\sqrt{N+1}\varphi(t)} + 
    \norm{k F_{\qh}} \norm{\sqrt{N+1}\varphi(t)} \norm{\ph \varphi(t)}\\
    &\lesssim (\norm{F}_2 + \norm{k F}_2)\norm{(\ph^2+N+1)^{1/2} \varphi(t)}^2 \\ & \lesssim \langle \varphi (t), (H_\hbar^0 +1 )\varphi (t)\rangle.
    \end{align*}
    By Gronwall lemma, for all $t \in \R$,
    $$\langle \varphi (t), H_\hbar^0 \varphi (t)\rangle \leq \langle \varphi (t), (H_\hbar^0+1) \varphi (t)\rangle \leq e^{\abs{t} C_\Lambda} \langle \varphi, (H_\hbar^0+1) \varphi \rangle.$$
    Taking $\varphi = \psi_\hbar (m)$, we obtain the wanted bound.

\end{proof}
\subsection{Duhamel formulas and commutator expansions}

\begin{prop}
Let $(\rho_\hbar)_\hbar \subset \mathfrak S^1 (\Gamma_s (\Hi))_{+,1}$ be a family of density matrices satisfying \eqref{assumption}. Then for any $t_0,t \in \R$, $\xi \in \Hi$, $\hbar \in (0,1)$, the following holds :
\begin{equation} \label{duhamelquantique}
\Tr ( W_\hbar (\xi)\tilde{\rho}_\hbar (t)) = \Tr ( W_\hbar (\xi)\tilde{\rho}_\hbar (t_0)) - \frac{i}{\hbar} \int_0^t \Tr \left( [W_\hbar (\xi), H_I(s)] \tilde{\rho}_\hbar (s) \right) \di s, 
\end{equation}
with
$$ 
H_I(s) = e^{\frac{i s g}{\hbar} H_\hbar^0} (H_{\hbar,g,F} - H_\hbar^0) e^{-\frac{i s g}{\hbar} H_\hbar^0} = e^{\frac{i s g}{\hbar} H_\hbar^0} \phi_\hbar(F_{\qh}) e^{-\frac{i s g}{\hbar} H_\hbar^0}. $$

\end{prop}

\begin{proof}
Since $[g H_\hbar^0, \qh^2 + \ph^2 + N_\hbar ] = 0$,  $e^{-\frac{i t}{\hbar} H_\hbar^0} (\qh^2 + \ph^2 + N_\hbar) e^{\frac{i t}{\hbar} H_\hbar^0} = \qh^2 + \ph^2 + N_\hbar$, hence 
$$\forall t \in \R, \forall \hbar \in (0,1), \Tr(\tilde{\rho}_\hbar (t) (\qh^2 + \ph^2 + N_\hbar)) = \Tr(\rho_\hbar (t) (\qh^2 + \ph^2 + N_\hbar)) \leq C e^{\abs{t} C}.$$
Adding the fact that $(H_\hbar^0+1)^{-1/2} W_\hbar (\xi) (H_\hbar^0+1)^{-1/2} $ is a bounded operator ensures that the map $t \longmapsto \Tr (\tilde{\rho}(t) W_\hbar (\xi))$ is differentiable, and the derivative has the wanted form.
\end{proof}

\begin{lem} For any $s \in \R$, we have on $Q(H_\hbar^0)$,
$$ H_I(s)  = \phi_\hbar(e^{i s g} F_{\qh + s g \ph}).$$
\end{lem}

\begin{proof}
    Noticing that $e^{\frac{i s g}{\hbar} H_\hbar^0}= e^{\frac{i s g}{\hbar} \frac{\ph^2}{2}} e^{\frac{i s g}{\hbar} N_\hbar} = e^{\frac{i s g}{\hbar} N_\hbar}  e^{\frac{i s g}{\hbar} \frac{\ph^2}{2}} $, a straigthforward computation shows that for all $\alpha \in \Hi_f$,
    \begin{align*}
    e^{\frac{i s g}{\hbar} \frac{\ph^2}{2}} \qh e^{-\frac{i s g}{\hbar} \frac{\ph^2}{2}} &= \qh + s g \ph, \\
    e^{\frac{i s g}{\hbar} N_\hbar} \phi_\hbar(\alpha) e^{-\frac{i s g}{\hbar} N_\hbar} &= \phi_\hbar(e^{i s g} \alpha).
    \end{align*}
    Hence,
    $$
    H_I(s) = e^{\frac{i s g}{\hbar} N_\hbar} \phi_\hbar(F_{\qh + s g \ph}) e^{-\frac{i s g}{\hbar} N_\hbar} = \phi_\hbar (e^{i s g} F_{\qh + s g \ph}).
    $$
\end{proof}

\begin{lem}
    For any $s \in \R$, we have on $Q(H_\hbar^0)$,
    $$
     W (\xi) H_I(s) W(\xi)^* - H_I(s) = \phi_\hbar (e^{i s g} (F_{\qh +s  g\ph -\hbar(q + s g p)} - F_{\qh + s g \ph})) - \hbar \mathrm{Im} \langle \alpha, e^{i s g} F_{\qh +s g \ph -\hbar(q + s g p )} \rangle.
    $$
    \end{lem}

\begin{proof}
    It is a straightforward consequence of the following formulas :
    \begin{align*}
        W_\hbar(\alpha) a_\hbar (\beta) W_\hbar (\alpha)^* &= a_\hbar (\beta) - i \frac{\hbar}{\sqrt{2}}  \langle \beta , \alpha \rangle_2, \\
        W_\hbar(\alpha) a_\hbar^* (\beta) W_\hbar (\alpha)^* &= a_\hbar^* (\beta) + i \frac{\hbar}{\sqrt{2}}  \langle \alpha , \beta \rangle_2,  \\
        T_\hbar (z) \ph_j T_\hbar (z)^* &= \ph_j - \hbar p_j,  \\
        T_\hbar (z) \qh_j T_\hbar (z)^* &= \qh_j - \hbar q_j.
    \end{align*}
\end{proof}

\begin{cor} For any $s \in \R$, we have on $Q(H_\hbar^0)$,\begin{equation}
 \frac{1}{\hbar} [W(\xi), H_I (s)] = (M(s,\xi) + \hbar  R(s,\xi,\hbar)) W(\xi), 
    \end{equation}
    with \begin{equation} \label{M(s,xi)}
    \begin{cases}
         M(s,\xi) &= \phi(e^{i s g } F_{\qh + s  g \ph} i k .(q+s g p)) - \mathrm{Im} \langle \alpha, e^{i s g} F_{\qh + s  g \ph} \rangle, \\
        R(s,\xi,\hbar) &= \phi \left( e^{i s g} F_{\qh + s g \ph} \frac{e^{i \hbar k .(q+s g p) } - 1 - i \hbar k .(q+s g p)}{\hbar^2} \right) - \mathrm{Im} \langle \alpha, e^{i s g } F_{\qh + s g \ph} \frac{e^{i \hbar k .(q+s g p)}-1}{\hbar} \rangle.
    \end{cases}
    \end{equation}
    For these quantities, we have the same estimate, uniform with respect to $\hbar$ :
    \begin{equation}
            \norm{(H_\hbar^0+1)^{-1/2} A (H_\hbar^0+1)^{-1/2}}_{\mathscr B(\Gamma_s (\Hi))} \lesssim_g (1+\abs{s})^2(1+\norm{\xi})^2,
    \end{equation}
    where $A \in \{ M (s, \xi) ,  R (s, \xi , \hbar) \}$.
\end{cor}

\begin{proof}
    The formulas \eqref{M(s,xi)} are immediate. For the estimates, it is a consequence of 
    $$
    \norm{(H_\hbar^0+1)^{-1/2} \phi_\hbar(\alpha) (H_\hbar^0+1)^{-1/2}}_{\mathscr B(\Gamma_s (\Hi))} \lesssim \norm{\alpha}_{L^2(\R^d)},
    $$
    for $\alpha \in \Hi_f$, as well as the bounds
    $
    \abs{\frac{e^{i \hbar x}-1}{\hbar}} \leq \abs{x}, \abs{ \frac{e^{i \hbar x}-1- i \hbar x}{\hbar^2}} \leq \frac{x^2}{2},
    $ holding for every $x \in \R$. We conclude by simply bounding 
    $$
    \abs{q + s p} \leq (1+\abs{s})(\abs{p}+\abs{q}) \lesssim (1+\abs{s})(1 + \norm{\xi}),
    $$
    and the scalar product by Cauchy-Schwarz inequality.
\end{proof}

\subsection{Derivation of the Liouville equation}

\begin{prop} \label{wignerexistence}
    Let $(\rho_\hbar)_\hbar \subset \mathfrak S^1 (\Gamma_s (\Hi))_{+,1}$ satisfying \eqref{assumption}. For any $\hbar_n \rightarrow 0$, there exists an extraction $(\hbar_{n_k})_{k \in \N} \rightarrow 0$ and a family $(\tilde{\mu}_t)_{t \in \R} \in \mathscr P(\Hi)^{\R}$ such that for every $t \in \R$,
    \begin{equation}
        \mathscr{M} (\tilde{\rho}_{\hbar_{n_k}} (t) , k \in \N) = \{ \tilde{\mu}_{t} \}.
    \end{equation}
    Furthermore, for all compact interval $J$,
    $$
    \underset{t \in J}{\sup} \int_{\Hi} \norm{\xi}^2 \di \tilde{\mu}_t (\xi) < \infty.
    $$
\end{prop}
Let $(t_j)_{j \in \N}$ be a dense countable family in $\R$. The propagation of regularity \eqref{propreg} ensures the existence for every $j \in \N$, of an extraction $\varphi_j : \N \rightarrow \N$ and Wigner measures $\tilde{\mu}_{t_j}\in \mathscr P(\Hi)$ such that $$\mathscr{M} (\tilde{\rho}_{\hbar_{\varphi_j(k)}} (t_j) , k \in \N) = \{ \tilde{\mu}_{t_j} \}.$$
By a diagonal argument, the extraction defined as $ \varphi : \N \rightarrow \N, k \mapsto \varphi_0 \circ \varphi_1 \circ ... \circ \varphi_k (k)$ yields a common subsequence $(\hbar_{\varphi(k)})_{k \in \N} \rightarrow 0$ along which for every $j \in \N$,
$$
\mathscr{M} (\tilde{\rho}_{\hbar_{\varphi(k)}} (t_j) , k \in \N) = \{ \tilde{\mu}_{t_j} \}.
$$
As detailed in \cite[Thm 6.2]{ammari2008ahp}, for any $j$ such that $t_j$ belongs to a compact interval $J$, there exists $C>0$ such that
\begin{equation}
\label{tightestimate}
\forall j \in \N, \int_{\Hi} \norm{\xi}^2 \di \tilde{\mu}_{t_j} (\xi) \leq C. 
\end{equation}
We define the characteristic functional of the family $\{ \tilde{\mu}_t \}_{t \in \R}$ at $\xi \in \Hi$ by
$$
G(t,\xi,\hbar) := \Tr ( W_\hbar (\xi)\tilde{\rho}_\hbar (t)).
$$
We have the following estimate :
\begin{lem} \label{uniformchar} Let $J$ be a compact interval. For all $t,s \in J$, for all $\xi, \eta \in \Hi$,
\begin{align*}
\abs{G(t,\xi,\hbar)-G(s,\xi,\hbar)} &\lesssim \abs{t-s} (1+\norm{\xi})^2, \\
\abs{G(t,\xi,\hbar)-G(t,\eta,\hbar)} &\lesssim \norm{\xi - \eta} (1+ \norm{\xi}+\norm{\eta}).
\end{align*}
\end{lem}
\begin{proof} \label{estimatechar}
From the Duhamel formula we have
\begin{align*}
    \abs{G(t,\xi,\hbar)-G(s,\xi,\hbar)} &= \abs{\int_s^t \Tr \left((M(\tau,\xi)+\hbar R(\tau,\xi)) W(\tilde{\xi}) \tilde{\rho}_\hbar (\tau) \right) \di \tau} \\
    &\leq \bigg{|} \int_s^t    \norm{(H_\hbar^0+1)^{-1/2}(M(\tau,\xi)+\hbar R(\tau,\xi))(H_\hbar^0+1)^{-1/2}}_{\mathscr B(\Gamma_s (\Hi))}  \\ & \ \ \ \ \ \ \ \ \ \times \norm{(H_\hbar^0+1)^{1/2} W(\tilde{\xi})(H_\hbar^0+1)^{-1/2}}_{\mathscr B(\Gamma_s (\Hi))} \\ & \ \ \ \ \ \ \ \ \ \times \Tr( (H_\hbar^0+1)^{1/2} \tilde{\rho}_\hbar (\tau) (H_\hbar^0+1)^{1/2})  \di \tau \bigg{|} \\
    &\lesssim \abs{ \int_s^t (1+\abs{\tau})^2 (1+\norm{\xi})^2 e^{\tau C} \di \tau} \\
    &\lesssim_J \abs{t-s} (1+\norm{\xi})^2.
\end{align*}
For the second assertion, write
\begin{align*}
    \abs{G(t,\xi,\hbar)-G(t,\eta,\hbar)} &= \abs{ \Tr( (W_\hbar(\xi) - W_\hbar(\eta)) \tilde{\rho}_\hbar (t))}
    \\
    &=\abs{ \Tr( (W_\hbar(\xi) - W_\hbar(\eta)) (H_\hbar^0+1)^{-1} (H_\hbar^0+1) \tilde{\rho}_\hbar (t))} \\
    &\leq \norm{(W_\hbar(\xi) - W_\hbar(\eta)) (H_\hbar^0+1)^{-1/2}} \Tr ( (H_\hbar^0+1) \tilde{\rho}_\hbar (t)) \\
    &\lesssim \norm{\xi - \eta} (1+ \norm{\xi}+\norm{\eta}),
\end{align*}
where we used \cite[Lemma 3.1]{ammari2008ahp} and \eqref{propreg} in the last inequality.
\end{proof}
\begin{proof}[Proof of Proposition \ref{wignerexistence}] Let $t \in \R$ and take a sequence of times $t_{\psi(j)} \rightarrow t$. For any $j \in \N$, the definition of Wigner measures ensures the existence of 
$$
G(t_{\psi(j)},\xi) := \underset{k \in \N}{\lim} \ G(t_{\psi(j)},\xi,\hbar_{\varphi(k)}).
$$

Lemma \ref{uniformchar} indicates that $(G(t_{\psi(j)},\xi))_{j \in \N}$ is a Cauchy sequence : indeed since $t_{\psi(j)} \rightarrow t$, $\{t_{\psi(j)} \}_{j \in \N}$ lies in a compact interval, we can take the limit $\hbar_{\varphi(k)} \rightarrow 0$ in \eqref{estimatechar},  with $t = t_{\psi(j)}$ and $s=t_{\psi(k)}$. We thus define
$$
G(t,\xi) := \underset{j \rightarrow \infty}{\lim} \  G(t_{\psi(j)},\xi).
$$
Now, we want to recognize $G(t,\cdot)$ as a characteristic functional of some probability measure over $\Hi$. 
It is immediate that $G(t,0) = 1$ and that $G(t,\cdot)$ is of positive type. Continuity is ensured by Lemma \ref{estimatechar}. By Bochner's theorem, there exists $\tilde{\mu}_t \in \mathscr P (\Hi)$ such that for all $\xi =(z,\alpha) \in \Hi$, if we denote $\tilde{\xi} = (\frac{z}{2 i \pi}, \frac{\alpha}{\sqrt{2} \pi})$, then
$$
G(t,\tilde{\xi}) = \int_{\Hi} e^{2 i \pi \mathrm{Re} \langle \xi, u \rangle} \di \tilde{\mu}_t (u).
$$
Now, \eqref{tightestimate} tells us that the family $\{\tilde \mu_{t_{\psi(j)}} \}_{j \in \N}$ is tight : there exists $\nu_t \in \mathscr P(\Hi)$ such that $\{\tilde \mu_{t_{\psi(j)}} \}_{j \in \N}$ is converging weakly narrowly to $\nu_t$. In fact, $\nu_t = \tilde{\mu}_t$. Indeed, 
$$
\int_{\Hi} e^{2 i \pi \mathrm{Re}\langle \xi, \cdot \rangle} \di \nu_t = \underset{j \rightarrow \infty}{\lim} \ \int_{\Hi} e^{2 i \pi \mathrm{Re}\langle \xi, \cdot \rangle} \di \tilde \mu_{t_{\psi(j)}} = \underset{j \rightarrow \infty}{\lim} \ G(t_{\psi(j)},\tilde \xi) = G(t,\tilde \xi) =  \int_{\Hi} e^{2 i \pi \mathrm{Re} \langle \xi, u \rangle} \di \tilde{\mu}_t (u).
$$
In consequence, $t \mapsto \tilde{\mu}_t$ is weakly narrowly continuous. We now compute
\begin{align*}
    \abs{G(t,\tilde{\xi},\hbar_{\varphi(k)}) - \int_{\Hi} e^{2 i \pi \mathrm{Re}\langle \xi, \cdot \rangle} \di \tilde{\mu}_t } &\leq \abs{G(t,\tilde \xi,\hbar_{\varphi(k)}) - G(t_{\psi(j)},\tilde \xi,\hbar_{\varphi(k)})} \\
    &+ \abs{G(t_{\psi(j)},\tilde \xi,\hbar_{\varphi(k)}) - \int_{\Hi} e^{2 i \pi \mathrm{Re}\langle \xi, \cdot \rangle} \di \tilde{\mu}_{t_{\psi(j)}}} \\
    &+ \abs{\int_{\Hi} e^{2 i \pi \mathrm{Re}\langle \xi, \cdot \rangle} \di \tilde{\mu}_{t_{\psi(j)}} - \int_{\Hi} e^{2 i \pi \mathrm{Re}\langle \xi, \cdot \rangle} \di \tilde{\mu}_{t}}.
\end{align*}
Hence,
\begin{align*}
\underset{k \in \N}{\limsup} \abs{G(t,\tilde \xi,\hbar_{\varphi(k)}) - \int_{\Hi} e^{2 i \pi \mathrm{Re}\langle \xi, \cdot \rangle} \di \tilde{\mu}_t }   &\lesssim \abs{t-t_{\psi(j)}} (1+\norm{\xi})^2 \\&+ \abs{\int_{\Hi} e^{2 i \pi \mathrm{Re}\langle \xi, \cdot \rangle} \di \tilde{\mu}_{t_{\psi(j)}} - \int_{\Hi} e^{2 i \pi \mathrm{Re}\langle \xi, \cdot \rangle} \di \tilde{\mu}_{t}}.
\end{align*}
Taking the limit $j \rightarrow \infty$ of the inquality we obtain that for every $t \in \R$ and $\xi \in \Hi$, 
$$
\underset{k \in \N}{\lim} \  G(t,\tilde \xi,\hbar_{\varphi(k)}) = \int_{\Hi} e^{2 i \pi \mathrm{Re}\langle \xi, \cdot \rangle} \di \tilde{\mu}_t = G(t,\xi).
$$
This is a reformulation of
$$
\mathscr{M} (\tilde{\rho}_{\hbar_{\varphi(k)}} (t) , k \in \N) = \{ \tilde{\mu}_{t} \}.
$$
The second assertion is again a consequence of \cite[Thm 6.2]{ammari2008ahp}.
\end{proof}

\subsection{Derivation of the characteristic equation}

In this subsection we take the limit in the Duhamel formula \eqref{duhamelquantique} to obtain an equation satisfied by $t \mapsto \tilde{\mu}_t$. The objective is to recognize $\tilde \mu_t$ as a pushforward of the initial measure by the interacting flow $\tilde \Phi_t$.

\begin{lem}
    With the same assumptions and notations as in Prop \ref{wignerexistence}, we have that for all $t \in \R$, for all $\xi = (z,\alpha) \in \Hi$,
    $$  \tilde{\mu}_t \left( e^{i \mathrm{Im} \langle  \cdot,z \rangle} e^{\sqrt{2} i \mathrm{Re} \langle \alpha, \cdot \rangle } \right) =  \tilde{\mu}_0 \left( e^{i \mathrm{Im} \langle  \cdot,z \rangle} e^{\sqrt{2} i \mathrm{Re} \langle \alpha, \cdot \rangle } \right) - i \ \underset{k \in \N}{\lim} \int_0^t \Tr \left(  \tilde M (s,\xi) W_{\hbar_{\varphi(k)}}(\xi) \tilde{\rho}_{\hbar_{\varphi(k)}} (s)  \right) \di s.
            $$
\end{lem}

\begin{proof}
    The result follows by directly taking the limit $(\hbar_{\varphi(k)})_{k \in \N} \rightarrow 0 $ in the Duhamel formulas, using the uniform estimates for $\tilde R(s, \xi, \hbar_{\varphi(k)})$ and $\hat R (\xi, \hbar_{\varphi(k)})$ to see that the latter vanishes at the limit.
\end{proof}

\begin{lem}
    For every $s \in \R$ and $\xi \in \Hi$,
    $$
    \underset{k \in \N}{\lim} \Tr \left( M (s,\xi) W_{\hbar_{\varphi(k)}}(\xi) \hat{\rho}_{\hbar_{\varphi(k)}} (s)  \right) = \tilde{\mu}_s \left( m(s, \xi, \cdot)  e^{i \mathrm{Im} \langle  \cdot,z \rangle} e^{\sqrt{2} i \mathrm{Re} \langle \alpha, \cdot \rangle }  \right), $$    where
    $$
        \tilde m (s,\xi , u) =  \sqrt{2} \mathrm{Re} \langle  \alpha, e^{i s g} F_{q+ s g p} i k . (q_0 + s g p_0) \rangle - \mathrm{Im} \langle \alpha_0, e^{i s g} F_{q+ s g p}  \rangle,  $$
\end{lem}

\begin{proof}
    The proof is a consequence of either \cite[Lemma C.1, Lemma C.2]{afh22} or \cite[Lemma C.1, Lemma C.2]{farhat2024}. We will only give ideas of the proof. The first step is to approximate the operator $M(s,\xi)$ by a new operator $M_n(s,\xi) := (m_n(s,\xi,\cdot))^{\mathrm{Wick}}$ emerging as the Wick quantization of a compact symbol, in a way that
    $$
    \norm{(N_\hbar+\qh^2 + \ph^2)^{-1/2} (M(s,\xi)-M_n(s,\xi))(N_\hbar+\qh^2 + \ph^2)^{-1/2}} \leq C_n(\xi) \underset{n\rightarrow \infty} \longrightarrow 0,
    $$
    where $C_n(\xi)$ does not depend on the semiclassical parameter $\hbar_{\varphi(k)}$. Thanks to \cite[Theorem 6.13]{ammari2008ahp}, the result holds when $M(s,\xi)$ and $m(s,\xi)$ are replaced with $M_n(s,\xi)$ and $m_n(s,\xi)$ respectively. Thanks to the uniform bound on the approximation as well as Proposition \ref{propreg}, the two quantities
    $$
     \Tr \left( (M (s,\xi)-M_n(s,\xi)) W_{\hbar_{\varphi(k)}}(\xi) \hat{\rho}_{\hbar_{\varphi(k)}} (s)  \right) \text{ and }
    \tilde{\mu}_s \left((m(s, \xi, \cdot)-m_n(s,\xi,\cdot))  e^{i \mathrm{Im} \langle  \cdot,z \rangle} e^{\sqrt{2} i \mathrm{Re} \langle \alpha, \cdot \rangle }  \right),
    $$
    can be made arbitrarily small as $n \rightarrow \infty$, uniformly in $\hbar_{\varphi(k)}$. Therefore it suffices to take $n$ large enough and then $\hbar_{\varphi(k)}$ small enough to conclude by an $\varepsilon/3$ argument.
\end{proof}

\begin{cor}
    For all $t \in \R$, for all $\xi = (z,\alpha) \in \Hi$,
$$ \tilde{\mu}_t \left( e^{i \mathrm{Im} \langle  \cdot,z \rangle} e^{\sqrt{2} i \mathrm{Re} \langle \alpha, \cdot \rangle } \right) =  \tilde{\mu}_0 \left( e^{i \mathrm{Im} \langle  \cdot,z \rangle} e^{\sqrt{2} i \mathrm{Re} \langle \alpha, \cdot \rangle } \right) - i \int_0^t \tilde{\mu}_s \left( m(s, \xi, \cdot)  e^{i \mathrm{Im} \langle  \cdot,z \rangle} e^{\sqrt{2} i \mathrm{Re} \langle \alpha, \cdot \rangle }  \right) \di s, $$
\end{cor}

We recall that the vector field $X : \R \times \Hi \rightarrow \R \times \Hi$ defined in the interaction picture is
\begin{equation}
\begin{aligned}
X(t,(q,p,\alpha)) :=& \Phi^0_{-t} \circ \mathcal N  \circ  \Phi^0_t (q,p,\alpha) \\
=& \begin{pmatrix} \sqrt{2} t \mathrm{Re} \langle e^{-i t g} \alpha , - i k F_{q+ t g p} \rangle_2 \\   -\sqrt{2} \mathrm{Re} \langle e^{-i t g} \alpha, - i k F_{q+ t g p} \rangle_2  \\  -i e^{i t g} \frac{1}{\sqrt{2}} F_{q+t g p} 
\end{pmatrix}.
\end{aligned}
\end{equation}

\begin{prop}
    With the same assumptions and notations as in Prop \ref{wignerexistence}, $t \mapsto \tilde{\mu}_t$ is a weakly narrowly continuous map satisfying the following characteristic equation : for all $t \in \R$, for all $y \in \Hi$,
    $$ \tilde{\mu}_t \left( e^{2 i \pi \mathrm{Re}\langle y, \cdot \rangle} \right) =  \tilde{\mu}_0 \left( e^{2 i \pi \mathrm{Re}\langle y, \cdot \rangle} \right) + 2 i \pi \int_0^t \tilde{\mu}_s \left( \mathrm{Re}\langle X(s,\cdot) , y \rangle e^{2 i \pi \mathrm{Re}\langle y, \cdot \rangle} \right) \di s.$$
    \end{prop}

\begin{proof} Let $\xi = (z_0, \alpha_0) \in \Hi$ and make the change of variable $y = (\tilde z_0, \tilde \alpha_0) = \left( \frac{z_0}{2 i \pi} , \frac{\alpha_0}{ \sqrt{2} \pi} \right)$. This means that $q_0 = - 2 \pi \tilde p_0$ and $p_0 = 2 \pi \tilde q_0,$. In consequence,
$$
 i \mathrm{Im}\langle \cdot, z_0 \rangle + \sqrt{2} i \mathrm{Re} \langle \alpha_0, \cdot \rangle = 2 i \pi \mathrm{Re} \langle y , \cdot \rangle.
$$
Hence the statement is equivalent to proving that
$$  m (s,\xi, \cdot) = - 2 \pi \mathrm{Re} \langle X(s,\cdot), y \rangle, $$
Indeed, for $u = (q,p,\alpha) \in \Hi$,
\begin{align*}
 m (s,\xi, \cdot) =& \sqrt{2} \mathrm{Re} \langle  \alpha, e^{i s g} F_{q+ s g p} i k . (q_0 + s g p_0) \rangle - \mathrm{Im} \langle \alpha_0, e^{i s g} F_{q+ s g p}  \rangle \\
=&  -2 \pi  \tilde q_0 . \left( - s g \sqrt{2} \mathrm{Re} \langle  \alpha, e^{i s g} F_{q+ s g p} i k \rangle \right) \\
 &-2 \pi \tilde p_0 . \left( \sqrt{2} \mathrm{Re} \langle  \alpha, e^{i s g} F_{q+ s g p} i k \rangle   \right)
\\
&-2 \pi \mathrm{Re} \langle \tilde \alpha_0, -i \frac{1}{\sqrt{2}} F_{q+s g p} \rangle,
\end{align*}
and we recognize the interacting flow.
\end{proof}

We now have all the tools to prove Theorem \ref{wignertheorem}. As $X$ and $t \mapsto \tilde \mu_t$ satisfies the assumptions of Proposition \ref{liouvilleequivalence}, we deduce that $t \mapsto \tilde \mu_t$ satisfies the Liouville equation : $\forall \phi \in  \mathscr C^\infty_{0,\text{cyl}} (I \times \Hi)$,
    $$
    \int_I \int_{\Hi} \left(  \partial_t \phi (t,u) + \mathrm{Re} \langle X(t,u) ,  \nabla_{\Hi} \phi (t,u)   \rangle_{\Hi}   \right) \di \tilde \mu_t (u) \di t = 0.
    $$
Furthermore, thanks to the uniform in time condition \eqref{uniformmeasurebound}, and the expression for $X$, it is clear that the vector field satisfies the integrability assumption
$$
t \mapsto \int_{\Hi} \norm{X(t,u)}_{\Hi} \di \tilde \mu_t (u) \di t \in L^1_{\mathrm{loc}}(\R).
$$
Therefore, by Proposition \ref{globalsuperposition}, we can construct $\nu \in \mathscr P (\Hi \times \mathscr C(\R,\Hi))$ such that
\begin{itemize}
    \item $\di \nu (x,\gamma)$ a.s., $\gamma$ is an integral solution of the initial value problem such that $\gamma(0) = x$.
    \item $\mu_t = (e_t)_* \nu$ for all $t \in \R$,
    \end{itemize}
where $e_t(x,\gamma) = \gamma(t)$ and $\mathscr C(\R,\Hi)$ is endowed with the compact open topology. Furthermore, thanks to the global well posedness on the classical side, we have that $\di \nu(x,\gamma)$ a.s., $\gamma(t) = (\tilde \Phi_t)_* \mu$. Therefore, for every smooth cylindrical function $\varphi \in \mathscr C^\infty_{0,\mathrm{cyl}} (\Hi)$, 
$$
\int_{\Hi} \varphi \di \tilde \mu_t = \int_{\Hi \times \mathscr C(\R,\Hi)} \varphi ( e_t (x,\gamma)) \di \nu(x,\gamma) = \int \varphi(\tilde \Phi_t (x)) \di \nu(x,\gamma) = \int \varphi(\tilde \Phi_t (x)) \di \mu (x),
$$
where in the last equality we used that the first marginal of $\nu$ is $\mu = \mu_0$, which is true since 
$$
\int \varphi (x) \di \mu_0(x) = \int \varphi(\gamma(0)) \di \nu(x,\gamma) = \int \varphi(x) \di \nu(x,\gamma).
$$
Therefore, we can conclude that $$ \tilde \mu_t = (\tilde \Phi_t)_* \mu. $$


\appendix

\section{Removal of the cutoff in the sense of quadratic forms}

\begin{prop} \label{lemmaA1}
    For any $\hbar \in (0,1)$, $H_\hbar^\infty$ is a well defined quadratic form with form domain $Q(H_\hbar^\infty) = Q(H_\hbar^0)$.
\end{prop}

\begin{proof}
    Let $\hbar \in (0,1)$. We will show that $\phi_\hbar(  f_{\qh}^\infty )$ is infinitesimally $H_\hbar^0$ bounded. Let $\psi \in Q(H_\hbar^0)$. For $K>0$, we split $f_{\qh}^\infty = \Id_{\abs{k}\leq K} f_{\qh}^\infty + \Id_{\abs{k}  > K} f_{\qh}^\infty.$ Since $\Id_{\abs{k}\leq K} f_{\qh}^\infty \in L^2$, we can write
\begin{align*}
 \abs{ \langle  \psi, \phi (\Id_{\abs{k}\leq K} f_{\qh} ) \psi  \rangle } & \leq \sqrt{2} \norm{\psi} \norm{\Id_{\abs{k}\leq K} f_{\qh}} \norm{\sqrt{N+1} \psi } \\ & \leq \eps \norm{(H_0+1)^{1/2} \psi}^2 + C_{\eps,K} \norm{\psi}^2 \\
 & \leq \eps \langle \psi , H_\hbar^0 \psi \rangle + C_{\eps,K} \norm{\psi}^2 \\
 &= \eps q_0 (\psi) + C_{\eps,K} \norm{\psi}^2.
\end{align*}

For the other side, first write
$$ \frac{e^{-i k.x}}{\abs{k}^{(d-1)/2}} = \frac{1}{\hbar \abs{k}^{(d+3)/2}} \sum_{j=1}^d k_j^2 e^{-i k . x}. $$
A straightforward computation allows to rewrite 
$$ \frac{e^{-i k.x}}{\abs{k}^{(d-1)/2}} = \frac{1}{\hbar \abs{k}^{(d+3)/2}} \sum_{j=1}^d k_j [P_j, e^{-i k . x}]. $$
And since the field operator commutes with $P_j$, for all $\psi \in D((H_\hbar^0)^{1/2})$ we have :
\begin{align*}
    \abs{\langle \psi , \phi \left( \frac{e^{-i k.x}}{\hbar \abs{k}^{(d-1)/2}} \Id_{\abs{k} > K} \right) \psi \rangle} & \leq \sum_{j=1}^d \abs{ \langle \psi , \left[P_j, \phi \left( \frac{k_j e^{-i k.x} \Id_{\abs{k} > K}}{\hbar \abs{k}^{(d+3)/2}} \right)  \right] \psi \rangle  } \\
    & \leq 2 \sum_{j=1}^d \norm{P_j \psi} \norm{\phi \left( \frac{k_j e^{-i k.x} \Id_{\abs{k} > K}}{\hbar \abs{k}^{(d+3)/2}} \right)  \psi }.
\end{align*}
Now, 
$\norm{P_j \psi}^2 = \langle \psi, P_j^2 \psi \rangle \leq \langle \psi, -\Delta \psi \rangle = q_0 (\psi),$
and 
$$ 
\norm{\phi \left( \frac{k_j e^{-i k.x} \Id_{\abs{k} > K}}{\hbar \abs{k}^{(d+3)/2}} \right)  \psi }^2 \leq \norm{\sqrt{N+1} \psi} \norm{ \frac{k_j \Id_{\abs{k} > K}}{\hbar \abs{k}^{(d+3)/2}}}_{L^2} \leq (q_0(\psi)+ \norm{\psi}^2) \norm{ \frac{ \Id_{\abs{k} > K}}{\hbar \abs{k}^{(d+1)/2}}}_{L^2}.
$$
Hence for $K_\eps$ large enough, 
$$
 \abs{\langle \psi , \phi \left( \frac{e^{-i k.x}}{\hbar \abs{k}^{(d-1)/2}} \Id_{\abs{k} > K} \right) \psi \rangle} \leq \eps q_0(\psi) + C_\eps \norm{\psi}^2,
$$
which concludes.
\end{proof}

\section{Equivalence between Liouville and characteristic equation}

The setting is the following : we consider  a separable Hilbert space $(\Hi,\langle \cdot,\cdot \rangle_{\Hi})$, and a continuous vector field $X: \R \times \Hi \rightarrow \Hi $ bounded on bounded sets. We are interested in the initial value problem
$$
\frac{\di u}{\di t} = X(t,u), \  u(0) = u_0 \in \Hi.
$$
Since in general $\Hi$ is infinite-dimensional, we need to define an appropriate set of test functions. This leads to the following definition :
\begin{defn}
    We say that $\phi : \Hi \rightarrow \C$ is a smooth cylindrical function whenever there exists $n \in \N$, $(e_1,...,e_n) \in \Hi^n$ and $\psi \in \mathscr C^\infty_0 (\R \times \R^n)$ such that for all $u \in \Hi$,
    $$
    \phi (t,u) = \psi (t, \mathrm{Re}\langle u, e_1 \rangle_{\Hi},...,\mathrm{Re}\langle u, e_n \rangle_{\Hi}).
    $$
    We denote by $\mathscr C^\infty_{0,\text{cyl}} (\R \times \Hi)$ the set of cylindrical functions.
\end{defn}

    Using the same notations as above, the gradient of a cylindrical test function $\phi \in \mathscr C^\infty_{0,\text{cyl}} (\R \times \Hi)$ can be expressed as
    $$
    \nabla_{\Hi} \phi (t,u) = \sum_{i=1}^n \partial_i \psi (t, \mathrm{Re}\langle u, e_1 \rangle_{\Hi},...,\mathrm{Re}\langle u, e_n \rangle_{\Hi}) e_i.
    $$

\begin{rem}
    Here, the gradient is computed using the real structure of $\Hi$, i.e. in the real Hilbert space $(\Hi, \mathrm{Re}\langle \cdot, \cdot \rangle_{\Hi})$.
\end{rem}

\begin{prop} \label{liouvilleequivalence}
Consider a continuous vector field $X: \R \times \Hi \rightarrow \Hi $, bounded on bounded sets. Let $t \in \R \mapsto \mu_t \in \mathscr P (\Hi)$ weakly narrowly continuous such that
$$
t \mapsto \int_{\Hi} \norm{X(t,u)}_{\Hi} \di \mu_t (u) \di t \in L^1_{\mathrm{loc}}(\R).
$$
Then for every interval $I$ whose interior contains $0$, there is an equivalence between :
\begin{enumerate}
    \item  $\{ \mu_t  \}_{t \in I }$ satisfies the Liouville equation : $\forall \phi \in  \mathscr C^\infty_{0,\text{cyl}} (I \times \Hi)$,
    $$
    \int_I \int_{\Hi} \left(  \partial_t \phi (t,u) + \mathrm{Re} \langle  X(t,u) ,  \nabla_{\Hi} \phi (t,u)   \rangle_{\Hi}   \right) \di \mu_t (u) \di t = 0.
    $$
    \item $\{ \mu_t \}_{t \in I}$ satisfies the characteristic equation : $\forall t \in I, \forall y \in \Hi$,
    $$
    \mu_t \left( e^{2 i \pi \mathrm{Re}\langle y, \cdot \rangle_{\Hi}  } \right) =  \mu_0 \left( e^{2 i \pi \mathrm{Re}\langle y, \cdot \rangle_{\Hi}} \right  ) + 2 i \pi \int_0^t \mu_s \left( \mathrm{Re}\langle X(s,\cdot) , y \rangle_{\Hi} e^{2 i \pi \mathrm{Re}\langle y, \cdot \rangle_{\Hi}} \right) \di s.
    $$
\end{enumerate}
\end{prop}

Let us now state a result obtained in \cite{AFS24}. We endow $\mathscr C (\R,\Hi)$ with the compact open topology, and denote $e_t : (x,\gamma) \in \Hi \times \mathscr (\R,\Hi) \mapsto \gamma(t) \in \Hi $.

\begin{prop} (global superposition principle) \label{globalsuperposition}
    Consider a continuous vector field $X: \R \times \Hi \rightarrow \Hi $, bounded on bounded sets. Let $t \in \R \mapsto \mu_t \in \mathscr P (\Hi)$ be a weakly narrowly continuous solution if the Liouville equation such that
$$
t \mapsto \int_{\Hi} \norm{X(t,u)}_{\Hi} \di \mu_t (u) \di t \in L^1_{\mathrm{loc}}(\R).
$$
Then there exists a probability measure $\nu \in \mathscr P (\Hi \times \mathscr C(\R,\Hi))$ such that
\begin{itemize}
    \item $\di \nu (x,\gamma)$ a.s., $\gamma$ is an integral solution of the initial value problem such that $\gamma(0) = x$.
    \item $\mu_t = (e_t)_* \nu$ for all $t \in \R$.
\end{itemize}
\end{prop}

\bibliographystyle{plain}

\end{document}